\documentclass[envcountsame,envcountsect,11pt,a4paper]{llncs}
\usepackage[top=30truemm,bottom=30truemm,left=25truemm,right=25truemm]{geometry}

\let\accentvec\vec  
\let\vec\accentvec 
\spnewtheorem*{Proof}{Proof}{\itshape}{\rmfamily}
\renewenvironment{proof}{\begin{Proof}}{\qed\end{Proof}}

\usepackage{amsmath,amssymb}

\usepackage{graphicx}

\usepackage[T1]{fontenc}
\usepackage[latin1]{inputenc}
\usepackage[english]{babel}

{\itshape}{\rmfamily}
\newtheorem{observation}[theorem]{Observation}

\usepackage{color}

\definecolor{lightblue}{rgb}{0.5,0.5,1.0}
\definecolor{darkred}{rgb}{0.6,0,0}
\definecolor{darkgreen}{rgb}{0,0.5,0}
\definecolor{darkblue}{rgb}{0,0,0.5}

\newcommand{\figref}[1]{\figurename~\ref{#1}}

\newcommand{\figdir}{.}

\DeclareMathOperator{\Mdeg}{Mdeg}
\DeclareMathOperator{\cw}{cw}
\DeclareMathOperator{\tw}{tw}
\DeclareMathOperator{\rw}{rw}

\title{Induced Minor Free Graphs: Isomorphism and Clique-width\thanks{An extended abstract of this paper previously appeared in the proceedings of the 41st International Workshop on Graph-Theoretic Concepts in Computer Science (WG 2015).
This work is partially supported by Grant-in-Aid for Young Scientists (B) 25730003, Japan Society for the Promotion of Science.}
}

\author{
R\'emy Belmonte\inst{1}
\and
Yota Otachi\inst{2}
\and
Pascal Schweitzer\inst{3}
}
\institute{
University of Electro-Communications, Tokyo, Japan\\  \texttt{remy.belmonte@gmail.com}  \and
Japan Advanced Institute of Science and Technology, Ishikawa, Japan\\  \texttt{otachi@jaist.ac.jp} \and
RWTH Aachen University, Aachen, Germany\\\texttt{schweitzer@informatik.rwth-aachen.de}
}

\usepackage{xspace}
\newcommand{\NPhard}{\CClassNP-hard\xspace}
\newcommand{\CClassNP}{\textup{NP}\xspace}

\begin{document}

\maketitle
\thispagestyle{plain}
\begin{abstract}
Given two graphs $G$ and $H$, we say that $G$ contains $H$ as an induced minor if a graph isomorphic to $H$ can be obtained from $G$ by a
sequence of vertex deletions and edge contractions. We study the complexity of {\sc Graph Isomorphism} on graphs that exclude a fixed graph
as an induced minor. More precisely, we determine for every graph~$H$ that {\sc Graph Isomorphism} is polynomial-time solvable
on~$H$-induced-minor-free graphs or that it is GI-complete.
Additionally, we classify those graphs~$H$ for which $H$-induced-minor-free graphs have bounded clique-width.
These two results complement similar dichotomies for graphs that exclude a fixed graph as an induced subgraph, minor, or subgraph.
\end{abstract}

\section{Introduction}

Remaining unresolved, the algorithmic problem {\sc Graph Isomorphism} persists as a fundamental graph theoretic challenge which, despite
generating ongoing interest, has neither been shown to be~\NPhard nor polynomial-time solvable.
It is known that if {\sc Graph Isomorphism} is NP-hard then the polynomial hierarchy collapses~\cite{Sch88}.
Recently, Babai has announced a quasipolynomial-time algorithm for {\sc Graph Isomorphism}~\cite{Babai15}.
The problem asks whether two given graphs
are structurally the same, that is, whether there exists an adjacency and non-adjacency preserving map from the vertices of one graph to the
vertices of the other graph.

{\it Related work.} In the absence of a result determining the complexity of the general problem, considerable effort has been put into
classifying the isomorphism problem restricted to graph classes as being polynomial-time tractable or polynomial-time equivalent to the
general problem, i.e., GI-complete. Most graph classes considered in these efforts are graph classes that are closed under some basic
operations. Operations that are typically considered are edge contraction, vertex deletion, and edge deletion. A class of graphs closed under
all of these operations is said to be minor closed and can also be described as a class of graphs avoiding a set of forbidden minors. As
shown by Ponomarenko, the {\sc Graph Isomorphism} problem can be solved in polynomial time on $H$-minor free graphs for any fixed graph
$H$~\cite{Ponomarenko88}. This implies prior results on solvability of graphs of bounded tree-width, planar graphs, and graphs of
bounded genus. The result on minor closed graph classes was recently extended by Grohe and Marx to $H$-topological minor free
graphs~\cite{GroheM15}, and Lokshtanov, Pilipczuk, Pilipczuk and Saurabh~\cite{LPPS14} showed that the problem is actually fixed-parameter
tractable on graphs of bounded tree-width, an important class of minor-free graphs. When a graph class is only required to be closed under
some of the above named operations, isomorphism on such a graph class can sometimes be polynomial-time solvable and sometimes be
GI-complete. We say that a graph $G$ is $H$-free if it does not contain the graph $H$ as an induced subgraph. When forbidding one induced
subgraph, it is known that {\sc Graph Isomorphism} can be solved in polynomial time on $H$-free graphs if $H$ is an induced subgraph
of~$P_4$ (the path on four vertices) and is GI-complete otherwise (see~\cite{BoothColbourn1979}).
For two forbidden induced subgraphs such a classification into GI-complete and polynomial-time solvable cases turns out to be
more complicated~\cite{DBLP:conf/wg/KratschS12,stacsSchweitzer}. In the case where we consider forbidden subgraphs (i.e., also allowing edge
deletions) there is a complete dichotomy for the computational complexity of {\sc Graph Isomorphism} on classes characterized by
a finite set of forbidden subgraphs, while there are intermediate classes defined by infinitely many forbidden subgraphs
for which the problem is neither polynomial-time solvable nor GI-complete~\cite{OtachiS13}
(assuming that graph isomorphism is not polynomial-time solvable).
Another related result is the polynomial-time isomorphism test for graphs of bounded clique-width recently developed~\cite{rankwidth}. We discuss the relationship to our results below.

{\it Our results.} In this paper we consider graph classes closed under edge contraction and vertex deletion (but not necessarily under edge
deletion). The corresponding graph containment relation is called induced minor. More precisely, a graph~$H$ is an \emph{induced minor} of a
graph~$G$ if~$H$ can be obtained from~$G$ by repeated vertex deletion and edge contraction. If no induced minor of~$G$ is isomorphic to~$H$,
we say that~$G$ is~\emph{$H$-induced-minor-free}. We consider graph classes characterized by one forbidden induced minor, and on these classes we
study the computational complexity of {\sc Graph Isomorphism} and whether the value of the parameter clique-width is bounded by
some universal constant $c_H$. The isomorphism problem for such classes was first considered by Ponomarenko~\cite{Ponomarenko88} for the
case where~$H$ is connected. In that paper two choices for the graph~$H$ play a crucial role, namely choosing~$H$ to be the gem and
choosing~$H$ to be co-$(P_{3} \cup 2 K_{1})$ (see Figure \ref{fig:important-graphs}).
Forbidding either of these graphs as an induced minor yields a graph class with an isomorphism problem solvable in polynomial time.
However, to show polynomial-time solvability for the gem, the proof of~\cite{Ponomarenko88}, due to a common misunderstanding concerning the
required preconditions, incorrectly relies on a technique of~\cite{DBLP:conf/coco/HopcroftT72} to reduce the problem to the 3-connected case
(see Subsection~\ref{subsection:gem}). To clarify the situation, we provide a proof that avoids this reduction and instead use a reduction
of the problem to the 2-connected case for which we provide a polynomial-time isomorphism test.
To extend Ponomarenko's theorem to the disconnected case, we provide a reduction structurally different from the ones used previously,
allowing us to treat the case where~$H$ consists of a cycle with an added isolated vertex.
Overall we extend Ponomarenko's results to obtain the following theorem.
\begin{figure}[htb]
  \centering
  \includegraphics[scale=.8]{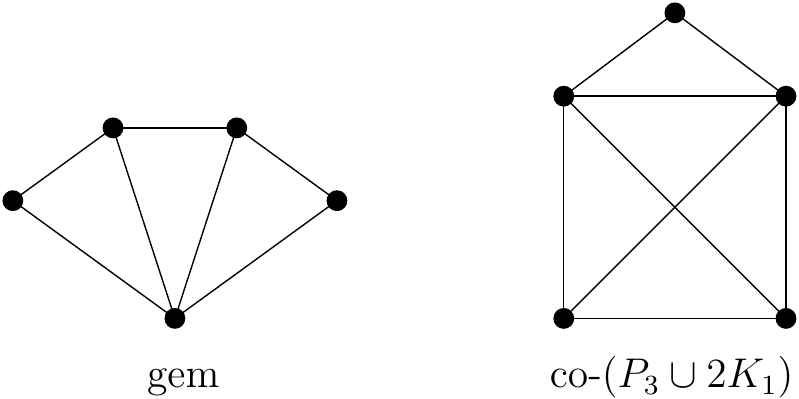}
  \caption{The graphs gen and co-$(P_{3} \cup 2 K_{1})$.}
  \label{fig:important-graphs}
\end{figure}

\begin{theorem}
Let $H$ be a graph. The {\sc Graph Isomorphism} problem on $H$-induced-minor-free graphs is polynomial-time solvable if $H$ is complete or
an induced subgraph of co-$(P_{3} \cup 2 K_{1})$ or the gem, and is GI-complete otherwise.
\end{theorem}

Our proofs rely on structural descriptions that also allow us to determine exactly which classes characterized by one forbidden induced
minor have bounded clique-width.

\begin{theorem}
\label{thm:clique-width}
Let $H$ be a graph. The clique-width of the $H$-induced-minor-free graphs is bounded if and only if $H$ is an induced subgraph of
co-$(P_{3} \cup 2 K_{1})$ or the gem.
\end{theorem}
Note that these two graphs play important roles also in a recent independent paper by B{\l}asiok et al.~\cite{BlasiokKRT15}.

As mentioned before, it was recently shown that {\sc Graph Isomorphism} is polynomial-time solvable for graphs of bounded clique-width~\cite{rankwidth}.
The proof of this theorem relies on the structure theory of connectivity functions, tangles
and, the computational group theory developed in the context of graph isomorphism.
For the various classes we consider, our algorithms do not rely on such machineries.  In fact, we give structural descriptions of the graphs
and, consequently, showing that the clique-width is bounded essentially amounts to the same arguments as those required to develop
polynomial-time isomorphism algorithms.
Note however that, by the theorems above, some classes for which we develop polynomial-time isomorphism algorithms have unbounded clique-width. 
Of course in general, it is not difficult to see that there are graph classes with unbounded clique-width on which {\sc Graph Isomorphism}
is polynomial-time solvable.  For example, planar graphs~\cite{DBLP:conf/coco/HopcroftT72}, interval graphs~\cite{LuekerB79},
and permutation graphs~\cite{Colbourn81} are such graph classes~\cite{GolumbicR00}.

Also note that $H$-free graphs have bounded clique-width if and only $H$ is an induced subgraph of $P_4$~\cite{DP16+} and that $H$-minor-free
graphs have bounded clique-width if and only if $H$ is planar~\cite{KaminskiLozinMilanic2009}.  Recently, Dabrowski and Paulusma gave a
dichotomy for the clique-width of bipartite $H$-free graphs~\cite{DP16}, and initiated the study of clique-width on graphs that forbid two
graphs as induced subgraphs~\cite{DP16+}.

{\it Structure of the paper.}
We first summarize well-known observations about induced-minor-free graphs, isomorphism and clique-width (Section~\ref{sec:basic:obs}).
We then consider classes that are characterized by one forbidden induced minor on at most five vertices~(Section~\ref{sec:complicated}). 
Finally we show that the observations of Sections~\ref{sec:basic:obs} and~\ref{sec:complicated} resolve all cases with forbidden induced minors on at least six vertices (Section~\ref{sec:at:least:6:vert}). 
In this paper all graphs that are considered are finite. 

\paragraph{Notation.}
For a graph~$G$ we denote by~$V(G)$ and~$E(G)$ the set of vertices and edges, respectively.
The \emph{neighborhood}~$N_G(v)$ of a vertex~$v\in V(G)$ is the set of vertices adjacent to~$v$.
We omit the index~$G$ if apparent from context.
For a subset of the vertices~$M\subseteq V(G)$ we denote by~$G[M]$ the subgraph of~$G$ \emph{induced} by~$M$.
The set~$N(M)$ is the set of vertices in~$V(G)\setminus M$ that have a neighbor in~$M$.
We write~$G-M$ for the graph~$G[V(G)-M]$ and~$\overline{G}$ for the edge complement of~$G$.
A connected graph is \emph{2-connected} if it has at least three vertices and
it remains connected after deleting any single vertex.
For two graphs $G_{1}$ and $G_{2}$ with $V(G_{1}) \cap V(G_{2}) = \emptyset$,
we denote by $G_{1} \cup G_{2}$ their \emph{disjoint union} $(V(G_{1}) \cup V(G_{2}), E(G_{1}) \cup E(G_{2}))$.
For example, the graph $K_{3} \cup K_{1}$ consists of a triangle and an isolated vertex.


\section{Basic observations}\label{sec:basic:obs}

In this section, we summarize a few well-known basic observations about clique-width and graph classes closed under induced minors.

\subsection{Clique-width}

In~\cite{CourcelleOlariu2000}, Courcelle and Olariu introduced the clique-width of graphs as a way of measuring the complexity of minimal
separators in a graph. Similarly to graphs of bounded tree-width, it has been shown that a large class of problems can be solved efficiently
on graphs of bounded clique-width~\cite{CMR00}.
It was only recently shown that {\sc Graph Isomorphism} is polynomial-time solvable for graphs of bounded clique-width~\cite{rankwidth}.

For any given graph $G$, the clique-width of $G$, denoted by $\cw(G)$, is defined as the minimum number of labels needed to construct $G$ by
means of the following 4 operations: (i) creation of a new vertex $v$ with label~$i$; (ii) forming the disjoint union of two labeled graphs
$G_1$ and $G_2$; (iii) joining by an edge every vertex labeled~$i$ to every vertex labeled~$j$, where~$i \neq j$; (iv) renaming label~$i$ to
label~$j$. In the remainder of the paper, we will use the following well-known observations to derive upper bounds or lower bounds on
the value of clique-width of $H$-induced-minor-free graphs. See e.g.,~\cite{HOSG08} for an overview of clique-width.

\begin{theorem}[\cite{CourcelleOlariu2000}]
\label{thm:cw-complement}
For every graph $G$, $\cw(G) \leq 2\cdot \cw(\overline{G})$ holds.
\end{theorem}

\begin{theorem}
[\cite{LozinRautenbach2004}] 
\label{thm:cw-vertex-deletion}
Let $G$ be a graph and $S$ a subset of the vertices of $G$. We have $\cw(G - S)\leq \cw(G) \leq 2^{|S|}(cw(G - S)+1) -1$.
\end{theorem}

Let $G$ be a graph and $u$ a vertex of $G$. The {\em local complementation} of $G$ at $u$ is the graph obtained from $G$ by replacing the subgraph induced by the neighbors of $u$ with its edge complement. The following observation follows from the well-known facts that for any graph $G$, we have $\rw(G) \le \cw(G) \leq 2^{\rw(G)+1}-1$ (see~\cite{OS06}), where~$\rw$ denotes the rank-width,  and that rank-width remains constant under local complementations~\cite{Oum05}.

\begin{observation}
\label{obs:cw-local-complementation}
Let $G$ and $G'$ be two graphs such that $G'$ can be obtained from $G$ by a sequence of local complementations, then $\cw(G) \leq 2^{\cw(G')+1}-1$.
\end{observation}

\begin{theorem}[\cite{Courcelle2014}]
\label{thm:cw-subdivision}
Let $G$ and $G'$ be two graphs such that $G'$ can be obtained from $G$ by a sequence of edge subdivisions, i.e., replacing edges with paths of length~2. Then $\cw(G) \leq 2^{\cw(G')+1}-1$.
\end{theorem}

\begin{theorem}[\cite{BoliacLozin2002,LozinRautenbach2004}]
\label{thm:cw-blocks}
Let $G$ be a graph and $\cal{B}$ the set of its 2-connected components. It holds that $\cw(G)\leq t + 2$, where $t=\max_{B\in\cal{B}}\{\cw(B)\}$.
\end{theorem}

Finally, note that for any graph $G$, the clique-width of $G$ is at most $3\cdot 2^{\tw(G)-1}$, where $\tw(G)$ denotes the tree-width of $G$~\cite{CR05}.

\subsection{Some tractable cases}
\begin{lemma}
If $H$ is a complete graph,
then {\sc Graph Isomorphism} for $H$-induced-minor-free graphs can be solved in polynomial time.
\end{lemma}
\begin{proof}
For any graph $H$,
{\sc Graph Isomorphism} for the $H$-minor-free graphs can be solved in polynomial time~\cite{Ponomarenko88}.
Since a graph has a complete graph $H$ as an induced minor
if and only if it has $H$ as a minor, the lemma follows.
\end{proof}

\begin{lemma}
Let $H$ be a complete graph $K_{k}$.
The $H$-induced-minor-free graphs have bounded clique-width if and only if $k \le 4$.
\end{lemma}
\begin{proof}
If $k \le 4$, then every $K_k$-induced-minor-free graph has tree-width at most~2.
Thus it has bounded clique-width~\cite{CR05,CourcelleOlariu2000}.
If~$k>4$, the set of $K_k$-induced-minor-free graphs includes all planar graphs.
Therefore the clique-width is unbounded (see \cite{KaminskiLozinMilanic2009}).
\end{proof}

Note that the lemma above is used to prove Theorem~\ref{thm:clique-width},
but $K_{4}$ is not explicitly mentioned in the statement, due to the fact that $K_{4}$ is an induced subgraph of co-$(P_3 \cup 2K_{1})$.

\begin{lemma}
\label{lem:P4-GI}
If $H$ is an induced subgraph of $P_{4}$
then {\sc Graph Isomorphism} for $H$-induced-minor-free graphs can be solved in linear time.
\end{lemma}
\begin{proof}
If $H$ is an induced subgraph of $P_{4}$,
then {\sc Graph Isomorphism} can be solved in linear time on $H$-free graphs (see {\cite[\S 2.7]{BoothColbourn1979}}).
Since any induced subgraph of $P_{4}$ is a linear forest,
a graph has $H$ as an induced minor if and only if
it has $H$ as an induced subgraph.
Therefore an $H$-induced-minor-free graph is $H$-free, and thus the lemma follows.
\end{proof}

It is well known that $P_{4}$-free graphs
are exactly the graphs of clique-width at most 2 (see \cite{KaminskiLozinMilanic2009}).

\subsection{Some intractable cases}\label{subsec:intractable:cases}
A \emph{split partition} $(C,I)$ of a graph $G$
is a partition of $V(G)$ into a clique $C$ and an independent set~$I$. A \emph{split graph} is a graph admitting a split partition.
We say a split graph is of \emph{restricted split type} if it has a split partition $(C, I)$
such that each vertex in $I$ has at most two neighbors in $C$.
Note that a non-complete split graph of restricted split type has minimum degree at most 2.
A graph is \emph{co-bipartite} if its vertex set can be partitioned into two cliques.
The classes of co-bipartite graphs and restricted split graphs are closed under vertex deletions and edge contractions, and thus under induced minors.
As also argued in~\cite{Ponomarenko88} and~\cite{DBLP:conf/wg/KratschS12},
the standard graph-isomorphism reductions to split graphs and co-bipartite graphs\footnote{%
First subdivide all edges, and then complement the color class corresponding to the original vertices (for split graphs of restricted type)
or complement both color classes (for co-bipartite graphs).}
explained in~\cite{BoothColbourn1979} imply the following lemma.
\begin{lemma}
\label{lem:type1split_cobipartite}
If $H$ is not of restricted split type or $H$ is not co-bipartite,
then {\sc Graph Isomorphism} for the $H$-induced-minor-free graphs is GI-complete.
\end{lemma}
The reductions used in the lemma can be achieved by performing edge subdivisions and subgraph complementation.
\emph{Subgraph complementation} is the operation of complementing the edges of an induced subgraph.
The clique-width of graphs in a class obtained by applying subgraph complementation a constant number of times
is bounded if and only if it is bounded for graphs in the original class~\cite{KaminskiLozinMilanic2009}.
Together with Theorem~\ref{thm:cw-subdivision}, this implies that restricted split graphs and co-bipartite graphs obtained by the reductions from general graphs have unbounded clique-width.
\begin{corollary}
\label{lem:type1split_cobipartite_cw}
If $H$ is not of restricted split type or $H$ is not co-bipartite,
then the $H$-induced-minor-free graphs have unbounded clique-width.
\end{corollary}


\section{Graphs on at most five vertices}
\label{sec:complicated}

In this section we study graph classes characterized by a forbidden induced minor~$H$ that has at most five vertices. 
In addition to the two graphs gem and co-$(P_{3} \cup 2 K_{1})$,
the graph $K_3 \cup K_1$ plays an important role here.
That is, all GI-complete and unbounded clique-with cases that cannot be handled by 
Lemmas~\ref{lem:type1split_cobipartite} and \ref{lem:type1split_cobipartite_cw}
can be handled by considering $K_3 \cup K_1$.
In the following, we first study the important graphs and then
reduce the remaining cases to the important cases.

\subsection{The graph $K_3 \cup K_1$}
We show that {\sc Graph Isomorphism} is GI-complete on graphs that do not contain $K_3 \cup K_1$ as an induced minor. Additionally, we show that these graphs have unbounded clique-width.

\begin{theorem}
\label{thm:K3uK1}
The {\sc Graph Isomorphism} problem is GI-complete on graphs that do not contain $K_3 \cup K_1$ as an induced minor.
\end{theorem}
\begin{proof}
We give a polynomial-time reduction from {\sc Graph Isomorphism} on graphs of minimum degree at least~3,
which is known to be GI-complete~\cite{BoothColbourn1979}.
Our reduction works as follows. Given two graphs $G$ and $H$ of minimum degree at least~3, we create two new graphs $G'$ and $H'$ by
subdividing every edge three times, i.e., we replace every edge with a path of length 4. From these graphs $G'$ and $H'$, we now create two
new graphs $G''$ and $H''$ by applying a local complementation at each vertex of $G'$ and $H'$ of degree at least~3, i.e., we turn the
neighborhood of every such vertex into a clique.  From the construction, it can readily be seen that $G$ is isomorphic to $H$ if and only if
$G''$ is isomorphic to $H''$. We now consider the complement graphs $\overline{G''}$ and $\overline{H''}$, which are isomorphic if and only
if $G''$ and $H''$ are isomorphic.
The graphs $\overline{G''}$ and $\overline{H''}$ can be obtained from $G$ and $H$ in polynomial time.

To conclude the proof, it suffices to show that $\overline{G''}$ and $\overline{H''}$ do not contain $K_3 \cup K_1$ as an induced minor.
Observe first that a graph $J$ does not contain $K_3 \cup K_1$ as an induced minor if and only if for every vertex $u \in V(J)$, the graph
$J - N_{J}[u]$ is a forest.  We claim that $\overline{G''}$ and $\overline{H''}$ satisfy this condition. We only give the proof for
$\overline{G''}$, since it is identical to that of $\overline{H''}$.
First, let us observe that the vertex set of $G''$ can be partitioned into three sets $S,T$, and $U$,
such that the vertices in $S = V(G)$ are simplicial, the vertices in $T$ are adjacent to exactly one vertex of $S$, and the vertices of $U$ have
 exactly two neighbors, both of which lie in~$T$. Let us first consider the case where $u$ is a vertex of $S$.
Since $u$ is simplicial in $G''$, its non-neighbors in $\overline{G''}$ are all pairwise non-adjacent,
so $\overline{G''} - N_{\overline{G''}}[u]$ has no edge.
Now, assume that $u$ is a vertex of $T$. Observe that exactly one neighbor $x$ of $u$ in $G''$ lies in $U$, and all its other neighbors lie
in either $S$ or $T$. Moreover, since $u$ is adjacent to a unique vertex $v$ of $S$, its neighbors in $T$ are also adjacent to $v$ and hence
form a clique in $G''$. This implies that $\overline{G''} - N_{\overline{G''}}[u]$ is a star centered at $x$.
Finally, assume that $u$ is a vertex of $U$. Since $u$ has exactly two non-neighbors in $\overline{G''}$,
it immediately follows that the non-neighbors of $u$ in $\overline{G''}$ induce a forest, thus completing the proof of the claim.
\end{proof}

\begin{theorem}
The class of graphs that do not contain $K_3 \cup K_1$ as an induced minor does not have bounded clique-width.
\end{theorem}
\begin{proof}
Assume for contradiction that there exists a constant $c \geq 2$ such that every graph that does not contain $K_3 \cup K_1$ as an induced
minor has clique-width at most $c$.
Let $G$ be a graph and let $\overline{G''}$ be the graph constructed from $G$ in the proof of Theorem~\ref{thm:K3uK1}. Observe that, as
noted in the proof of Theorem~\ref{thm:K3uK1}, $G''$ can be obtained from $G$ by a sequence of edge subdivisions and local
complementations. By applying Observation~\ref{obs:cw-local-complementation} and Theorem~\ref{thm:cw-subdivision}, we conclude that the
clique-width of~$G$ is bounded by a function of the clique-width of~${G''}$. Together with Theorem~\ref{thm:cw-complement}, this implies that
the clique-width of~$G$ is bounded by a function of the clique-width of~$\overline{G''}$. By choosing~$G$ such that its clique-width is
sufficiently large, we find that $\overline{G''}$ has clique-width at least $c+1$ and is $(K_3 \cup K_1)$-induced-minor-free, a contradiction.
\end{proof}

\subsection{The gem}\label{subsection:gem}

We now consider the class of graphs that do not contain the gem as an induced minor (see
\figref{fig:important-graphs}). In~\cite{Ponomarenko88} this class is also considered, however, there is an issue with the proof for the
fact that the isomorphism problem of graphs in this class is polynomial-time solvable. More precisely, a common misunderstanding of how the
reduction to 3-connected components by Hopcroft and Tarjan~\cite{DBLP:conf/coco/HopcroftT72} is to be applied has happened.
Indeed, the techniques of Hopcroft and Tarjan do not show that graph
isomorphism in a graph class~$\mathcal{C}$ polynomial-time reduces to
graph isomorphism of 3-connected components in~$\mathcal{C}$, even
if~$\mathcal{C}$ is induced minor closed. If this were the case
then the class of split graphs of restricted type would be
polynomial-time solvable since the only 3-connected graphs of this type
are complete graphs. Additionally to~$\mathcal{C}$ being induced minor closed,
for the techniques to be applicable it is necessary to solve the
edge-colored isomorphism problem for 3-connected graphs
in~$\mathcal{C}$.  However, edge-colored isomorphism is already
GI-complete on complete graphs.

We now provide a proof that isomorphism of graphs not containing the gem as an induced minor is polynomial-time solvable without alluding to 3-connectivity. For this we first need to extend the structural considerations for such graphs performed in~\cite{Ponomarenko88} for 3-connected graphs to 2-connected graphs.

Let~$C$ be a subset of the vertices of~$G$.
We say a vertex~$v$ in a vertex set~$M\subseteq V(G)\setminus C$ has \emph{exclusive attachment} with respect to~$C$ among the vertices of~$M$ if~$N(v)\cap C\neq \emptyset$ but there is no vertex~$v' \in M\setminus \{v\}$ with~$(N(v)\cap C) \cap (N(v')\cap C) \neq \emptyset$. That is, no other vertex of~$M$ shares a neighbor in~$C$ with~$v$.

\begin{lemma}
\label{only:one:neighbor:implies:exclusive:attachment}
Let~$G$ be a 2-connected gem-induced-minor-free graph.
Suppose~$C\subseteq V(G)$ induces a 2-connected subgraph of~$G$ and~$M$ is the vertices of a component
of~$G-C$ such that~$N(M) \cap C \neq C$.
If~$v\in M$ is a vertex with~$|N(v)\cap C| = 1$ then~$v$ has exclusive attachment.
\end{lemma}
\begin{proof}
Let~$v$ be a vertex in~$M$ that is adjacent to a vertex~$c_1\in C$ but not to any other vertex of~$C$. We argue that in this case~$v$ is the
only vertex in~$M$ adjacent to~$c_1$. Suppose otherwise that~$v'$ is a second vertex in~$M$ adjacent to~$c_1$.
Since~$G$ is 2-connected there is a path~$P$ from~$v'$ to~$C$ that does not use~$c_1$.
If~$v'$ is adjacent to some other vertex of~$C$ then we can ensure that~$P$ does not contain~$v$.
If~$v'$ is also only adjacent to~$c_1$ and no other vertex of~$C$ then~$v$ and~$v'$ are interchangeable.
By possibly swapping~$v$ and~$v'$, we can thus assume without loss of generality that
there is a path~$P$ from $v'$ to $C$ that contains neither $v$ nor $c_{1}$.
Since~$M$ induces a connected component of~$G-C$ there is a path~$P'$ in~$G[M]$ from~$v$ to~$v'$.
Consider a shortest path~$\widehat{P}$ starting from~$v$ that uses only vertices of~$(V(P) \cup V(P') \cup C)\setminus \{c_1\}$ and contains exactly two vertices of $C$.
Such a path exists since we can walk from~$v$ to~$v'$, then walk to~$C$ without passing~$c_1$ and then walk to a vertex
in~$C\setminus N(M)$ without using~$c_1$ due to 2-connectivity of~$C$.
Let~$c_{2}$ and $c_3$ be the vertices of $C$ on~$\widehat{P}$.
 
We claim that there is a vertex~$w$ in~$V(\widehat{P})\setminus C$ different from~$v$ that can be reached from~$c_1$ without using any other
vertex of~$\widehat{P}$. Indeed, since~$c_1$ is the only neighbor of~$v$ on~$C$, there must be some vertex in~$V(\widehat{P})\setminus C$
other than~$v$. Since~$c_1$ is adjacent to~$v'$ and all vertices of~$P$ and~$P'$ not on~$C$ other than~$v$ can be reached from~$v'$ without
using~$v$ or vertices from~$C$ there must be some vertex in~$\widehat{P}\setminus C$ that is reachable from~$v'$ (possibly~$v'$ itself
if~$v'\in \widehat{P}$) and thus reachable from~$c_1$.

The path~$\widehat{P}$ contains the set~$S= \{v,w,c_2,c_3\}$ of four distinct vertices (here~$w$ may or may not be~$v'$). We claim that
from~$c_1$ we can reach any vertex of~$S$ without traversing another vertex of~$\widehat{P}$. Indeed,~$c_1$ is adjacent
to~$v$. Furthermore~$c_1,c_2$ and~$c_3$ are three vertices in the 2-connected subgraph~$G[C]$ but no other vertex of~$\widehat{P}$ is in~$C$ and
from~$c_1$ we can reach~$w$ since~$w$ was chosen this way. This demonstrates the existence of a gem as an induced minor, yielding a
contradiction.
\end{proof}

\begin{lemma}
\label{lem:only:3:vertex:attached}
Let~$G$ be a 2-connected gem-induced-minor-free graph.
Suppose~$C\subseteq V(G)$ induces a 2-connected subgraph of~$G$
and~$M$ is the vertices of a component of~$G-C$ with~$N(M) \cap C \neq C$ and~$|N(M)\cap C|\leq 3$.
If there is no vertex~$x$ in~$M$ with~$|N(x)\cap C| = 1$ then every vertex of~$M$ has a neighbor in~$C$, and~$G[M]$ is a~$P_4$-free graph.
\end{lemma}
\begin{proof}
If there is no vertex~$x$ in~$M$ with~$|N(x)\cap C| = 1$
then, since~$|N(M)\cap C|\leq 3$, every pair of vertices in $M$ with neighbors in~$C$ has a common neighbor in~$C$. 
We show that every vertex in~$M$ has a neighbor in~$C$. Suppose there is a vertex~$u$ in~$M$ of distance~2 from~$C$. Let~$v$ be a neighbor
of~$u$ at distance~1 from~$C$. Since~$G$ is 2-connected, there is a second vertex~$v'$ in~$M$ at distance~1 from~$C$ such that there is a
path in~$G[M]$ from~$u$ to~$v'$ not using~$v$.
The vertex~$v'$ shares a neighbor~$c$ with~$v$. 
Consider the graph $G'$ obtained by contracting~$\{v,c\}$ and let $c'$ denote the vertex obtained from the contraction
and $C'$ denote the set $(C\setminus \{c\}) \cup c'$.
Let~$P$ be a path from $u$ to $v'$ in~$G'-C'$. Observe that $G'[C'\cup V(P)]$ is 2-connected. This can be seen from the fact that $G'[C']$ is a
spanning supergraph of $G[C]$, i.e., it can be obtained from $G[C]$ by adding edges, and that the endpoints of $P$, i.e., $u$ and $v'$, have
at least two distinct neighbors in $C'$.
The vertex~$u$ has exactly one neighbor in~$C'$, namely~$c'$, but~$v'$ also has~$c'$ as a neighbor.
By using Lemma~\ref{only:one:neighbor:implies:exclusive:attachment},
we obtain a contradiction to the assumption that $G$ does not contain a gem as
an induced minor. We conclude that no vertex in~$M$ is of distance at least~$2$ from~$C$, which implies that every vertex in~$M$ has a
neighbor in~$C$. If $G[M]$ contains an induced subgraph isomorphic to $P_4$, we can contract $C$ to a single vertex and obtain a gem as an
induced minor of $G$, again a contradiction. Therefore, we obtain that~$G[M]$ is~$P_4$-free, which completes the proof.
\end{proof}

We call a vertex of a 2-connected graph~$G$ a \emph{branching vertex} if it has degree at least 3.

\begin{lemma}
\label{lem:P4:or:cycle}
Let~$G$ be a 2-connected gem-induced-minor-free graph that contains the path~$P_4$ as an induced subgraph. Then at least one of the following two options holds:
\begin{itemize}
\item $G$ has an induced path~$H$ such that at most two of its inner vertices are branching vertices of~$G$ and~$G-V(H)$ is disconnected, or 
\item $G$ has an induced cycle~$H$ containing at most three branching vertices of~$G$ such that for every connected component of~$G-V(H)$ its vertex set~$M$ satisfies~$N(M)\cap H \neq H$.
\end{itemize}
\end{lemma}
\begin{proof}
Let~$K$ be an induced subgraph of~$G$ that is isomorphic to~$P_4$. Suppose that~$v_1,v_2,v_3,v_4$ are the vertices of the path encountered in that order.
If~$G-V(K)$ is disconnected, we are in the first case and there is nothing to show.
Thus, assume there is a unique component of~$G-V(K)$ induced by~$M$.
Since $G$ is 2-connected, $v_1$ and $v_4$ must each have a neighbor in $M$.
Since $G$ is gem-induced-minor-free, we may therefore assume that some inner vertex of~$K$, say~$v_2$, does not have a neighbor in $M$.
Since~$G$ is 2-connected, there must be a path from~$v_1$ to~$v_3$ avoiding~$v_2$. Let~$P$ be such a path of shortest length. If~$P$
contains less than two inner vertices that are branching vertices then we can add~$v_2$ to~$P$ and obtain a cycle with the desired properties.
Otherwise~$P$ contains a subpath~$P'$ with exactly two inner vertices that are branching vertices.
In this case~$G-V(P')$ is disconnected since otherwise contracting~$G-V(P')$ would yield the forbidden
induced minor. We choose~$H$ to be~$P'$ and obtain a path with the desired properties.
\end{proof}

Let~$G$ be a graph with induced subgraphs~$H$ and~$K$. We say that~$G$ is \emph{sutured} from~$H$ and~$K$ along~$V \subseteq V(H)$ and~$V'\subseteq V(K)$ if~$G$ is obtained in the following way. First we require that~$|V| = |V'|$.
We also require that~$V(H) \cap V(K) = V\cap V'$.
The graph~$G$ must then be formed from the (not necessarily disjoint) union~$(V(H) \cup V(K), E(H) \cup E(K))$ of $H$ and $K$ in the following way.
We add edges that form a perfect matching between vertices in~$V \setminus V'$ and~$V'\setminus V$. Finally we may subdivide the edges in
the matching an arbitrary number of times, see Figure~\ref{fig:suture}.

\begin{figure}[htb]
  \centering
  \includegraphics[scale=.6]{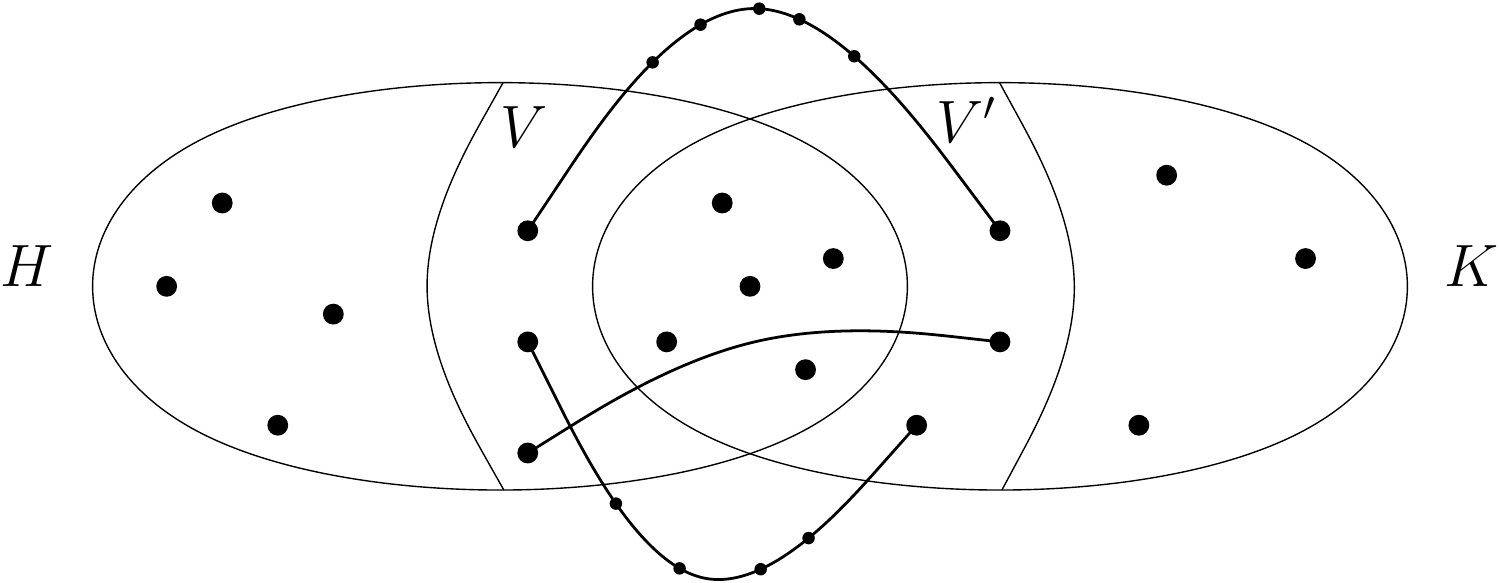}
  \caption{A suture of two graphs~$H$ and~$K$.}
  \label{fig:suture}
\end{figure}

\begin{lemma}
\label{lem:cycle:or:path:with:nice:attachemenst}
Let~$G$ be a 2-connected gem-induced-minor-free graph. There exists an induced subgraph~$H$ of~$G$ which is isomorphic to either a path or a cycle, contains at most~$4$ branching vertices,
and such that for every component of~$G-H$ its vertex set $M$ satisfies the following:
the graph~$G[M \cup V(H)]$ is sutured from~$H$ and some graph~$K$ along~$V$ and~$V'$ such that~$K - V'$ is~$P_4$-free and $|V'|\leq 4$. 
\end{lemma}
\begin{proof}
If~$G$ is~$P_4$-free then the lemma follows by choosing~$H$ to be the empty graph.
We thus assume that~$G$ contains an induced~$P_4$.
We distinguish the two cases that appear in Lemma~\ref{lem:P4:or:cycle}.

Suppose first that there is a graph~$H$ that is a path with at most
two inner branching vertices such that~$G-V(H)$ is disconnected. Let~$M$ be the vertices of a  connected component of~$G-H$.
Due to the forbidden induced minor~$|N(M) \cap V(H)| \leq 3$.
Since there is a second component of~$G-H$ and~$G$ is 2-connected there is a cycle in~$G - M$ induced by a vertex set $C$ such that~$N(M) \cap C \ne C$.
Let~$\widehat{M}$ be the set of vertices obtained from~$M$ by adding to~$M$ all vertices of~$V(H)\setminus C$ that lie on a subpath in~$H$ whose endpoints are both in~$N(M)\cap V(H)$.
By applying Lemma~\ref{only:one:neighbor:implies:exclusive:attachment}, we conclude that every vertex~$v\in \widehat{M}$ that is attached to only one vertex~$h\in C$ has
exclusive attachment.
Now, by repeatedly contracting the edges~$(v,h)$ of exclusive attachment we end up with either
no vertices in $\widehat{M}$ anymore, or all vertices in $\widehat{M}$ have at least two neighbors in~$C$.
Due to Lemma~\ref{lem:only:3:vertex:attached} we conclude that the graph induced by the vertices remaining in~$\widehat{M}$ is~$P_4$-free.
This gives us a suture of~$H[V(H) \cap C]$ with a~$P_4$-free graph, by increasing~$H[V(H) \cap C]$ to~$H$ we obtain a suture of~$H$ with a~$P_4$-free graph.

The second case appearing in Lemma~\ref{lem:P4:or:cycle} is similar. If~$C$ induces a cycle with at most~$3$ branching vertices such that for
every component of~$G-C$ its vertex set~$M$ satisfies $N(M)\cap C \neq C$, then again by Lemma~\ref{only:one:neighbor:implies:exclusive:attachment} we
conclude that vertices in~$\widehat{M}$ that are attached to only one vertex have exclusive attachment. Repeatedly contracting the exclusive
attachment edges and finally applying Lemma~\ref{lem:only:3:vertex:attached} we obtain a graph that is~$P_4$-free.

This shows that~$G[M \cup V(H)]$ is a suture of~$H$ with a graph~$K$ such that~$K - V'$ is~$P_4$-free. Since~$H$ has at most four
branching vertices, there are at most four points of attachment.
\end{proof}

\begin{theorem}
\label{thm:gem:iso}
The {\sc Graph Isomorphism} problem can be solved in polynomial time on gem-induced-minor-free graphs.
\end{theorem}

\begin{proof}
It is folklore that graph isomorphism in a hereditary graph class~$\mathcal{C}$ reduces to isomorphism of vertex-colored 2-connected graphs in~$\mathcal{C}$ (see for example~\cite{DBLP:conf/coco/DattaLNTW09} or~\cite{OtachiS14}).
We thus assume that the input graphs are colored and 2-connected. 
If~$G$ is such a 2-connected graph, we search for an induced subgraph~$H$ that satisfies the assumptions of
Lemma~\ref{lem:cycle:or:path:with:nice:attachemenst}, that is,~$H$ is a
path or a cycle with at most~$4$ branching vertices such that for every
component of~$G-H$ with vertex set $M$ we know that~$G[M \cup V(H)]$ is a
suture of~$H$
with a graph~$K$ along sets $V$ and $V'$ such that~$K - V'$ is~$P_4$-free and~$|V'|\leq 4$.
Each~$H$ is determined by the branching vertices, the leaves (if~$H$ is a path) and
choices of the paths of non-branching vertices connecting such vertices.
Note that there are only a linear number of such paths between a pair of vertices as the paths are internally disjoint.
 
Now suppose~$G_1$ and~$G_2$ are 2-connected input graphs to the isomorphism problem.
Since there are only polynomially many possible choices for $H$,
we can find an induced graph~$H_1$ in~$G_1$ with said properties and test for every~$H_2$ in~$G_2$ whether there is an isomorphism that maps~$H_1$ to~$H_2$.
To do so we iterate over all isomorphisms~$\varphi$ from~$H_1$ to~$H_2$,
there are only polynomially many, and check whether such an isomorphism extends to an isomorphism from~$G_1$ to~$G_2$.
To check whether such an isomorphism extends,
it suffices to know which component~$M_1$ of~$G_1-H_1$ can be mapped isomorphically to which component~$M_2$ of~$G_2-H_2$
such that the isomorphism can be extended to an isomorphism from~$G_1[V(H_1) \cup V(M_1)]$ to~$G_2[V(H_2)\cup V(M_2)]$
such that~$H_1$ is mapped to~$H_2$ in agreement with~$\varphi$.

Note that the mapping~$\varphi$ determines how vertices with exclusive attachment in~$H_1$ must be mapped. Letting~$A_1$ be the set of
vertices in~$M_1$ with exclusive attachment in~$V(H_1)$, we then know where the vertices in~$A_1$ must be mapped
if~$\varphi$ can be extended to an isomorphism from~$M_1$ to~$M_2$. Considering in turn the vertices of exclusive attachment in~$V(H_1)\cup A_1$ we obtain a set~$A_2$ of vertices for the images of which there is again only one possible option.
Repeating this process we obtain a sequence of sets~$A_1,\ldots, A_t$ such that there are no vertices in~$M_1-( A_1\cup\ldots\cup A_t)$
that have exclusive attachment. (The set~$A_1\cup\ldots\cup A_t$ contains the set~$V'\setminus V$ if~$V'$ is the set along which~$M_1$
is sutured to~$H_1$.)
We are left with~$M_1 -( A_1\cup\ldots\cup A_t)$, a part of~$M_1$ that is~$P_4$-free and adjacent to at most four vertices in~$V(H_1)\cup
A_1\cup\ldots\cup A_t$ whose images have already been determined.

The isomorphism problem for vertex-colored~$P_4$-free graphs is solvable in polynomial time (see~\cite{stacsSchweitzer}) and thus the problem for graphs obtained from~$P_4$-free graphs by adding a bounded number of vertices can be solved in polynomial time (\cite[Theorem 1]{DBLP:conf/swat/KratschS10}).
Using this algorithm the theorem follows.
\end{proof}

\begin{theorem}
\label{thm:gem_cw}
If $H$ is an induced subgraph of the gem,
then the $H$-induced-minor-free graphs have bounded clique-width.
\end{theorem}
\begin{proof}
Let~$G$ be a gem-induced-minor-free graph. Due to Theorem~\ref{thm:cw-blocks}, it suffices to show that the 2-connected components of $G$ have bounded clique-width.
By Lemma~\ref{lem:cycle:or:path:with:nice:attachemenst}, there is an induced subgraph~$H$ of~$G$ that is a path or a cycle with at most four branching vertices such that for every component
of~$G-V(H)$ induced by $M$ we know that~$G[M \cup V(H)]$ is a suture of~$H$ with a graph~$K$ such that~$K - V'$ is~$P_4$-free, where~$V'$ are the attachments in~$K$ and $|V'|\leq 4$. 
Therefore, each component of $G-H$ can be obtained from a disjoint union of a $P_4$-free graph and some paths
by adding at most four vertices, and there is a set $S$ of at
most four vertices of $H$, the branching vertices, such that each connected component of $G-S$ is either a connected component of $G-H$, or an
induced path. Therefore, each 2-connected component of $G$ can be obtained from a graph of bounded clique-width by adding at most~4
vertices. By using Theorem~\ref{thm:cw-vertex-deletion}, we obtain that each 2-connected component of $G$ has bounded clique-width.
\end{proof}

\subsection{The graph co-$(P_{3} \cup 2 K_{1})$}

In the following we will analyze the graphs that do not contain an induced minor isomorphic to~co-$(P_{3} \cup 2 K_{1})$, the graph obtained from~$K_5$ by removing two incident edges. While it has already been shown in~\cite{Ponomarenko88} that isomorphism for such graphs reduces to isomorphism of graphs not containing the gem (and is thus polynomially solvable), we provide a refinement of the proof in~\cite{Ponomarenko88} for this. We do this to obtain a finer structural description of these graphs, allowing us to also bound the clique-width in the graph class.

Suppose~$G$ is a~co-$(P_{3} \cup 2 K_{1})$-induced-minor-free graph. If~$G$ does not have a~$K_t$ minor for some fixed~$t$ then~$G$ is in particular in the minor closed graph class of~$K_t$-minor free graphs, and, as described in the introduction, the isomorphism problem can be solved in polynomial time for such graphs. 
Our strategy is thus to find a~$K_t$ minor and use this to analyze the structure of~$G$. In general, of course, there is no constant bound
on the number of vertices required to form a~$K_t$ minor. However in a~co-$(P_{3} \cup 2 K_{1})$-induced-minor-free graph there is such a
bound. We call a~$K_t$ minor \emph{compact} if every bag has at most two vertices.

\begin{lemma}\label{thm:existence:of:compact:minors}
If a co-$(P_{3} \cup 2 K_{1})$-induced-minor-free graph~$G$ has a~$K_t$ minor for~$t\geq 5$ then~$G$ has a compact~$K_t$ minor.
\end{lemma}

\begin{proof}
Let vertex sets~$M_1,\ldots,M_t$ be the bags of a~$K_t$ minor in~$G$ such that $M_i$ are inclusion minimal with respect to forming a~$K_t$ minor. That is, removing a vertex from one of the~$M_i$ yields a minor different from~$K_t$.
We analyze the structure of the minor. We say a vertex~$v$ is adjacent to a bag~$M_j$ if there exists a vertex~$v'\in M_j$ that is adjacent to~$v$.

For a vertex~$v \in M_i$ define~$\Mdeg(v) = |\{M_j \mid j\neq i, N(v) \cap M_j \neq \emptyset\}|$ to be the number of bags different from~$M_i$ adjacent to~$v$.
Using several steps we will show that~$\Mdeg(v) \geq t-2$ for all~$v \in M_1\cup M_2\cup \cdots\cup M_t$.
We first argue that if~$\Mdeg(v) > 1$ then~$\Mdeg(v) \geq t-2$.
Indeed, if~$\Mdeg(v) > 1$ then consider the minor obtained by removing all vertices from~$M_i$ different from~$v$. If~$\Mdeg(v) <t-2$ we can
choose two bags which have vertices adjacent to~$v$ and two bags which do not have such vertices. Using these bags and the vertex~$v$ we
obtain the forbidden induced minor~co-$(P_{3} \cup 2 K_{1})$.
We call vertices with~$\Mdeg(v) = 0$ \emph{inner vertices}, those with~$\Mdeg(v)= 1$ \emph{low degree vertices} and we call vertices with~$\Mdeg(v) \geq t-2$ \emph{high degree vertices}.
Next we argue that there are at most two high degree vertices in each bag.
First, observe that if $M_i$ contains a vertex $v$ such that $\Mdeg(v)=t-1$, then $v$ is the only vertex in $M_i$ and we are done. Therefore
we may assume that every vertex $v$ in $M_i$ satisfies $\Mdeg(v) \leq t-2$.
Now, if $M_{i}$ contains a high degree vertex, then we can pick two vertices~$v,v'$ in~$M_i$,
such that $v$ is a high degree vertex, $v'$ is adjacent to the bag that $v$ is not adjacent to, and there
is a path from~$v$ to~$v'$ in~$M_i$ that does not contain any other high degree vertex.
Since every bag different from~$M_i$ is adjacent to~$v$ or~$v'$, removing all vertices different from~$v$ and~$v'$ and not lying on the path
yields a~$K_t$ minor. Since the bags~$M_1,\ldots,M_t$ were chosen to be minimal, we conclude that there are at most two high degree vertices
in each bag.
 
We further argue that there is no low degree vertex in~$M_i$.
Indeed, if there is at least one low degree vertex in~$M_i$, we can choose a low degree vertex~$v\in M_i$ and a vertex~$v'\in M_i$
adjacent to a bag~$M_j$ with~$j\neq i$ such that~$v$ is not adjacent to~$M_j$ and such that there exists a path in~$M_i$ of inner vertices
connecting~$v$ and~$v'$. We remove all vertices in~$M_i$ different from~$v$ and~$v'$ and not on said path connecting them. We then move the
vertex~$v'$ from~$M_i$ to~$M_j$. We obtain the induced minor~co-$(K_{1,t-3} \cup 2 K_{1})$, which contains~co-$(P_{3} \cup 2 K_{1})$
since~$t \ge 5$.

Finally we argue that there are no inner vertices. Indeed, by minimality we can assume that every inner vertex~$v$ lies on a path between two high degree vertices~$v_1$ and~$v_2$, say. We again remove all vertices different from~$v_1$ and~$v_2$ not on the path. We then move~$v_1$ to an adjacent bag~$M_j$ and~$v_2$ to an adjacent bag~$M_{j'}$ such that~$j\neq j'$. This is possible since the vertices have high degree. Again we obtain a forbidden induced minor~co-$(K_{1,t-3} \cup 2 K_{1})$ as above.

Since there are only high degree vertices and since each bag can only contain two such vertices, the minimal minor is compact.
\end{proof}

\begin{lemma}
\label{lem:bicon:aug:co:graph}
If~$G$ is a 2-connected co-$(P_{3} \cup 2 K_{1})$ induced-minor-free graph and~$M$ is a compact~$K_t$ minor with~$t\geq 5$ then~$G-V(M)$ is $(K_{2} \cup K_{1})$-free. 
\end{lemma}
\begin{proof}
Assume that there is a 2-connected graph~$G$ that does not fulfill the lemma.
Let~$M_1,\ldots,M_t$ be the bags of the compact minor~$M$. Without loss of generality we choose the bags to be minimal with respect to inclusion.
Let~$v$ be a vertex in~$G-V(M)$. We argue that~$v$ is adjacent to all but at most one of the bags (in the terminology of the previous proof~$v$ is of high degree).

Suppose there are two bags in which $v$ does not have neighbors.
Since~$G$ is 2-connected, there are two vertex disjoint paths that start in~$v$ and end in distinct vertices of~$M$. We choose these paths to be
shortest among all possible choices, so if $v$ has neighbors in $M_i$, we pick the corresponding paths. Also, if $v$ has neighbors in
distinct bags $M_i$ and $M_j$, then pick the corresponding paths.
If these paths end in different bags~$M_j$ and~$M_k$ of~$M$ then we obtain the forbidden induced minor by moving all
vertices of each path into the bag in which the path ends and removing all other vertices of~$G-V(M)$ besides~$v$. If the paths end in the same
bag~$M_i$ of~$M$ we do the same operation by considering the minor~$M' = M-M_i$ and extending each path by one vertex. For this, note that
if~$u$ and~$u'$ are the vertices of~$M_i$ then we can choose~$M_j$ adjacent to~$u$ with~$j\neq i$ and~$M_k$ adjacent to~$u'$ with~$j\neq i$
and~$j\neq k$.
Since~$t\geq 5$ we still obtain the desired forbidden induced minor.

Suppose now that the vertices~$a,b,c$ induce the subgraph~$K_{2} \cup K_{1}$ in~$G-V(M)$.
Since~$t\geq 5$, there are bags~$M_{i}$ and~$M_{j}$ such that each vertex in~$\{a,b,c\}$ is adjacent to~$M_1$ and to~$M_2$. This implies that the family of bags~$\{a\}, \{b\}, \{c\},M_i, M_j$ induces the forbidden minor. This shows that~$G-V(M)$ is~$(K_{2} \cup K_{1})$-free.
\end{proof}

\begin{corollary}
\label{cor:minor-free}
If a 2-connected co-$(P_{3} \cup 2 K_{1})$-induced-minor-free graph~$G$ has a~$K_8$ minor then~$G$ is $(K_{2} \cup K_{1})$-free.
\end{corollary}
\begin{proof}
Assume that~$G$ has a~$K_8$ minor. 
Suppose~$G$ has an induced subgraph~$H$ isomorphic to~$K_{2} \cup K_{1}$. Since~$H$ has only~$3$ vertices,
$G-V(H)$ has a~$K_{8-3} = K_5$ minor. By Lemma~\ref{thm:existence:of:compact:minors} the graph~$G-V(H)$ has a compact~$K_5$ minor~$K$.
Thus by Lemma~\ref{lem:bicon:aug:co:graph} the graph~$G-V(K)$, which contains~$H$, is $(K_{2} \cup K_{1})$-free, yielding a contradiction.
\end{proof}

Since the gem is 2-connected, and thus every occurrence of a gem as an induced minor must occur within a 2-connected component of a graph, the
corollary is a refinement of Ponomarenko's result~\cite{Ponomarenko88} that says that if a co-$(P_{3} \cup 2 K_{1})$-induced-minor-free
graph~$G$ has a~$K_{2^{18}+4}$-minor then it does not contain a gem as an induced minor.

\begin{theorem}
\label{thm:ger:house:iso}
The Graph isomorphism problem for co-$(P_{3} \cup 2 K_{1})$-induced-minor-free graphs can be solved in polynomial time.
\end{theorem}
\begin{proof}
As mentioned in the proof of Theorem~\ref{thm:gem:iso}, graph isomorphism in a hereditary graph class~$\mathcal{C}$ reduces to
isomorphism of vertex-colored 2-connected graphs in~$\mathcal{C}$ (see for example~\cite{DBLP:conf/coco/DattaLNTW09} or~\cite{OtachiS14}).
We thus assume that the input graphs are colored and 2-connected. Also graph isomorphism is polynomial-time solvable for colored graphs in
non-trivial minor free graph classes~\cite{Ponomarenko88} so we can assume that both input graphs contain a~$K_8$ minor. By the previous
corollary we conclude that the input graphs are $(K_{2} \cup K_{1})$-free. Since isomorphism of colored cographs is solvable in polynomial
time, we obtain the theorem.
\end{proof}

To show that the co-$(P_{3} \cup 2 K_{1})$-induced-minor-free graphs have bounded clique-width, we need the following fact, which was indirectly proven by van 't Hof et al. in the proof of Theorem~9 in \cite{VantHofKaminskiPaulusmaSzeiderThilikos2012}.

\begin{theorem}
\label{lem:tw->indmin}
For any graph $F$ and for any planar graph $H$, there exists a constant $c_{F, H}$ such that an $F$-minor-free graph of tree-width
at least $c_{F, H}$ has $H$ as an induced minor.
\end{theorem}
\begin{proof}
Let $c_{F}$ be a constant such that every $F$-minor-free graph of tree-width
at least $c_{F} \cdot k^{2}$ has $\Gamma_{k}$ as an induced minor,
where $\Gamma_{k}$ is a planar graph of tree-width at least $k$~\cite{FominGT09}.
Let $b_{H}$ be a constant such that every planar graph of tree-width at least $b_{H}$
contains $H$ as an induced minor~\cite{FKMP95}.
Let $G$ be an $F$-minor-free graph of tree-width at least $c_{F} \cdot b_{H}^{2}$.
By the definition of $c_{F}$, $G$ has $\Gamma_{b_{H}}$ as an induced minor.
Now by the definition of $b_H$, $\Gamma_{b_{H}}$ has $H$ as an induced minor.
Thus $G$ has $H$ as an induced minor.
\end{proof}

\begin{theorem}
\label{thm:german-house_cw}
If $H$ is an induced subgraph of co-$(P_{3} \cup 2 K_{1})$,
then the $H$-induced-minor-free graphs have bounded clique-width.
\end{theorem}
\begin{proof}
Due to Theorem~\ref{thm:cw-blocks}, it suffices to show that the 2-connected components of $G$ have bounded clique-width.
Let $G$ be a 2-connected co-$(P_{3} \cup 2 K_{1})$-induced-minor-free graph.

First, assume that $G$ contains $K_{8}$ as a minor.
As a consequence of Corollary~\ref{cor:minor-free}, $G$ is~$(K_{2} \cup K_{1})$-free and therefore a cograph, implying that~$G$ has clique-width at most~2.

Next, assume that $G$ is $K_{8}$-minor-free. Since co-$(P_3 \cup 2K_1)$ is planar, there is a constant $c$ such that if a graph is $K_{8}$-minor-free and has tree-width at least $c$, it contains co-$(P_3 \cup 2K_1)$ as an induced minor, due to  Theorem~\ref{lem:tw->indmin}. Thus $G$ has tree-width at most $c$.
\end{proof}


\subsection{The remaining graphs on at most five vertices}
Now we study the remaining small graphs of at most five vertices.
We show that every case here can be reduced to some case we have already solved.
\begin{lemma}
Let $H$ be a non-complete graph on five vertices.
If $H$ is neither co-$(P_{3} \cup 2 K_{1})$ 
nor the gem,
then {\sc Graph Isomorphism} for $H$-induced-minor-free graphs is GI-complete.
\end{lemma}
\begin{proof}
By Lemma~\ref{lem:type1split_cobipartite} and Theorem~\ref{thm:K3uK1},
we may assume that $H$ is co-bipartite, a split graph of restricted split type, and $(K_{3} \cup K_{1})$-free.
Let~$K$ be a maximum clique of~$H$. Observe that $|K| \in \{3, 4\}$.

We first consider the case where $|K| = 4$.
Let $(C, I)$ be a restricted split partition of $H$.
Since every vertex in $I$ has degree at most 2, it holds that $K \subseteq C$.
This implies that $C = K$ and $|I| = 1$.
Since $H$ is $(K_{3} \cup K_{1})$-free, the vertex in $I$ has at least two neighbors in $C$.
Thus the vertex in $I$ has degree exactly two.
Now it holds that $H = \textrm{co-}(P_{3} \cup 2 K_{1})$.
(See \figref{fig:important-graphs}.)

Next we consider the case where $|K| = 3$.
Let $(C, I)$ be a restricted split partition of $H$.
Since $H$ is co-bipartite, we have $|C| = 3$ and $|I| = 2$.
Since $H$ is $(K_{3} \cup K_{1})$-free, each vertex in $I$ has at least one neighbor in $C$.
Since $H$ is co-bipartite, we can also assume that each vertex in $C$ has at least one neighbor in $I$.
On the other hand, since the maximum size of a clique in~$H$ is~$3$, no vertex in $I$ has three neighbors in $C$.
Therefore, the two vertices in $I$ have either both degree~2 or one of them is of degree~1 and the other of degree~2.
In the former case $H$ is the gem, and in the latter case, $H$ is the kite, which contains $K_{3} \cup K_{1}$.
(See \figref{fig:gem-kite}.)
\begin{figure}[htb]
  \centering
  \includegraphics[scale=.8]{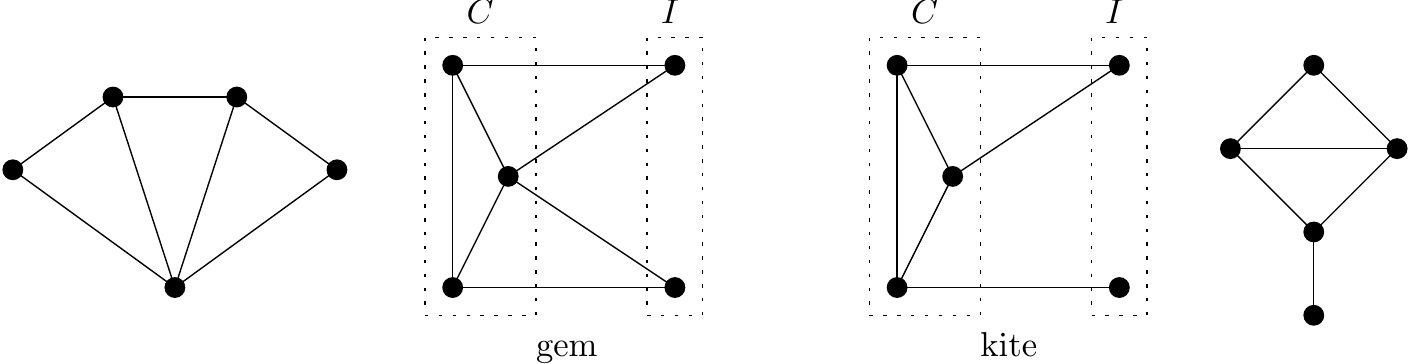}
  \caption{The graphs gem and kite and split partitions of these graphs.}
  \label{fig:gem-kite}
\end{figure}
\end{proof}

\begin{lemma}
Let $H$ be a graph on at most four vertices.
The {\sc Graph Isomorphism} problem for $H$-induced-minor-free graphs is polynomial-time solvable
if $H$ is an induced subgraph of either co-$(P_{3} \cup 2 K_{1})$ 
or $P_{4}$. Otherwise, it is GI-complete.
\end{lemma}
\begin{proof}
By Lemma~\ref{lem:type1split_cobipartite}, we may assume that $H$ is co-bipartite and split,
since otherwise the problem is GI-complete and
$H$ is not an induced subgraph of co-$(P_{3} \cup 2 K_{1})$ or $P_{4}$.
There are 11 non-isomorphic graphs on four vertices~\cite{ISGCI-small_graphs}.
It is easy to check that only five of them, depicted in \figref{fig:4vers}, are co-bipartite and split.
One of the five is $K_{3} \cup K_{1}$ and the others are $P_{4}$ and three induced subgraphs of co-$(P_{3} \cup 2 K_{1})$.
For graphs on at most three vertices,
we can easily check that all co-bipartite split graphs are
induced subgraphs of co-$(P_{3} \cup 2 K_{1})$.
By Lemma~\ref{lem:P4-GI} and Theorems~\ref{thm:K3uK1} and \ref{thm:ger:house:iso},
the lemma follows.
\begin{figure}[htb]
  \centering
  \includegraphics[scale=.8]{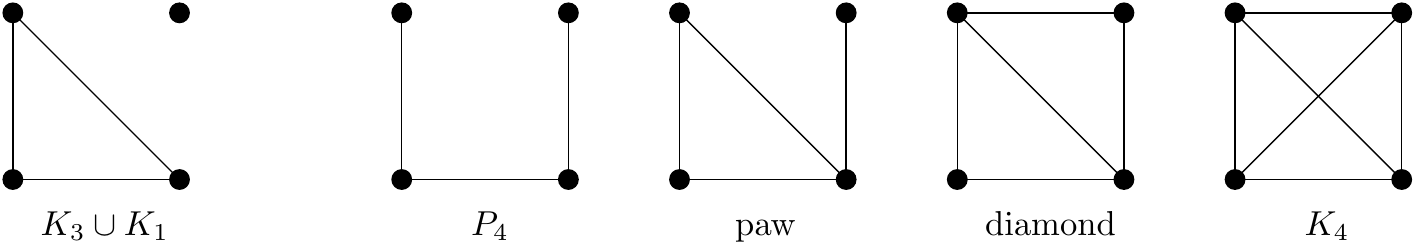}
  \caption{Co-bipartite split graphs on four vertices.}
  \label{fig:4vers}
\end{figure}
\end{proof}

The two lemmas above together imply the following theorem.
\begin{theorem}
Let $H$ be a non-complete graph on at most five vertices.
Then {\sc Graph Isomorphism} for $H$-induced-minor-free graphs
is polynomial-time solvable if $H$ is an induced subgraph of
co-$(P_{3} \cup 2 K_{1})$ or the gem;
otherwise, it is GI-complete.
\end{theorem}

The reductions we used above in order to show GI-completeness preserve
the property that the clique-width is unbounded (see Subsection~\ref{subsec:intractable:cases}). Thus we have the following corollary.
\begin{corollary}
Let $H$ be a non-complete graph on at most five vertices.
Then the $H$-induced-minor-free graphs have bounded clique-width
if and only if $H$ is an induced subgraph of co-$(P_{3} \cup 2 K_{1})$ or the gem.
\end{corollary}


\section{Non-complete graphs on at least six vertices}\label{sec:at:least:6:vert}

In this section, we show that if $H$ is not a complete graph and has at least six vertices,
then {\sc Graph Isomorphism} for the $H$-induced-minor-free graphs is GI-complete.
\begin{lemma}
\label{lem:K5}
If $H$ is non-complete and contains a clique of size 5,
then {\sc Graph Isomorphism} for $H$-induced-minor-free graphs is GI-complete.
\end{lemma}
\begin{proof}
Let $K$ be a clique of size 5 in $H$.
By Theorem~\ref{thm:K3uK1}, we may assume that $H$ does not have $K_{3} \cup K_{1}$ as an induced subgraph.
This implies that each vertex of $H$ has at least three neighbors in $K$, and thus $\delta(H) \ge 3$.
Since $H$ is non-complete, in any split partition $(C, I)$ of $H$, the independent set $I$  is non-empty.
These two facts together imply that $H$ is not of restricted split type.
Therefore, by Lemma~\ref{lem:type1split_cobipartite}, the lemma follows.
\end{proof}

\begin{theorem}
If $H$ is a non-complete graph on at least six vertices,
then {\sc Graph Isomorphism} for $H$-induced-minor-free graphs is GI-complete.
\end{theorem}
\begin{proof}
By Lemma~\ref{lem:type1split_cobipartite},
we may assume that $H$ is co-bipartite and of restricted split type.
By Lemma~\ref{lem:K5}, we may also assume that $H$ has no clique of size 5.
A split graph with seven or more vertices has a clique of size 5 or an independent set of size~3.
Thus we may assume that $|V(H)| = 6$.

Observe that $H$ has a clique $K$ of size 4,
since otherwise it contains an independent set of size 3.
Let $(C, I)$ be a restricted split partition of $H$.
Since every vertex in $I$ has degree at most 2, it holds that $K \subseteq C$.
Observe that no vertex in $V(H) \setminus K$ can be in $C$, since $H$ has no clique of size 5. 
This implies that $I = V(H) \setminus K$ and $C = K$. Note that $|I| = 2$ and $|C| = 4$.

Since $H$ has no independent set of size 3,
every vertex in $C$ has a neighbor in~$I$.
On the other hand the vertices in $I$ have degree at most 2.
Therefore, $H$ is the graph obtained from $K_{4}$ by adding two vertices of degree 2
so that the new vertices have no common neighbor. (See \figref{fig:coH}.)
Because $H$ contains $K_{3} \cup K_{1}$, the theorem follows by Theorem~\ref{thm:K3uK1}.
\begin{figure}[htb]
  \centering
  \includegraphics[scale=.8]{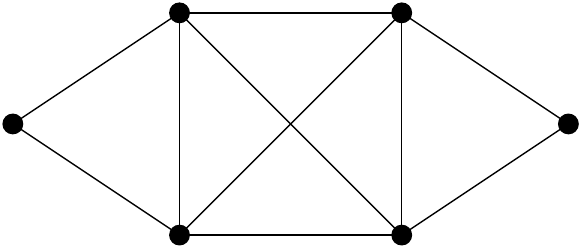}
  \caption{The graph co-$\mathsf{H}$.}
  \label{fig:coH}
\end{figure}
\end{proof}

Since the reductions that we used above in order to show GI-completeness preserve
the property that the clique-width is unbounded (see Subsection~\ref{subsec:intractable:cases}), we have the following corollary.
\begin{corollary}
If $H$ is a non-complete graph on at least six vertices,
then $H$-induced-minor-free graphs have unbounded clique-width.
\end{corollary}


\section*{Acknowledgments}
The authors thank the anonymous reviewers for constructive comments that improved the presentation of the paper.


\bibliographystyle{plain}
\bibliography{indmin}

\end{document}